\documentclass[leqno]{amsart}
\usepackage{amsmath,amsthm,amsfonts,amssymb}

\usepackage{natbib}

\numberwithin{equation}{section}
\theoremstyle{plain}

\newtheorem{assumption}{Assumption}
\newtheorem{theorem}{Theorem}

\newtheorem{lemma}[theorem]{Lemma}

\newtheorem{proposition}[theorem]{Proposition}

\newtheorem{remark}[theorem]{Remark}

\begin{document}

\title[Risk Processes with Subexponential Claims]
{Ruin Probabilities for Risk Processes with Non-Stationary Arrivals and Subexponential Claims}

\author{LINGJIONG ZHU}
\address
{Courant Institute of Mathematical Sciences\newline
\indent New York University\newline
\indent 251 Mercer Street\newline
\indent New York, NY-10012\newline
\indent United States of America}
\email{ling@cims.nyu.edu}

\date{6 April 2013. \textit{Revised:} 3 August 2013}
 %risk theory & insurance; point processes; large deviations 

\begin{abstract}
In this paper, we obtain the finite-horizon and infinite-horizon 
ruin probability asymptotics for risk processes with
claims of subexponential tails for non-stationary arrival processes 
that satisfy a large deviation principle.
As a result, the arrival process can be dependent, non-stationary and non-renewal. 
We give three examples of non-stationary and non-renewal point processes: Hawkes process, Cox process with shot
noise intensity and self-correcting point process. 
We also show some aggregate claims results for these three examples.
\end{abstract}

\maketitle

\section{Introduction}

Let us consider a classical risk model
\begin{equation}
U_{t}=u+pt-\sum_{i=1}^{N_{t}}C_{i},
\end{equation}
where $C_{i}$ are i.i.d. claims distributed as an $R^{+}$-valued random variable $C$, $p>0$ is the premium rate, 
$u>0$ is the initial reserve and $N_{t}$ is a simple point process.

We are interested in the case when $C_{i}$ have heavy tails.
A distribution function $B$ is subexponential, i.e. $B\in\mathcal{S}$ if 
\begin{equation}
\lim_{x\rightarrow\infty}\frac{\mathbb{P}(C_{1}+C_{2}>x)}{\mathbb{P}(C_{1}>x)}=2,
\end{equation}
where $C_{1}$, $C_{2}$ are i.i.d. random variables with distribution function $B$. 
Let us denote $B(x):=\mathbb{P}(C_{1}\geq x)$
and let us assume that $\mathbb{E}[C_{1}]<\infty$ and define $B_{0}(x):=\frac{1}{\mathbb{E}[C]}\int_{0}^{x}\overline{B}(y)dy$,
where $\overline{F}(x)=1-F(x)$ is the complement of any distribution function $F(x)$.
In the paper, the notation $f(x)\sim g(x)$ means $\lim_{x\rightarrow\infty}\frac{f(x)}{g(x)}=1$.

\cite{Goldie} showed that if $B\in\mathcal{S}$ and satisfies some smoothness
conditions, then $B$ belongs to the maximum domain of attraction of either the Frechet distribution
or the Gumbel distribution. In the former case, $\overline{B}$ is regularly varying,
i.e. $\overline{B}(x)=L(x)/x^{\alpha+1}$, for some $\alpha>0$ and we write
it as $\overline{B}\in\mathcal{R}(-\alpha-1)$, $\alpha>0$, where $L(\cdot):\mathbb{R}^{+}\rightarrow\mathbb{R}^{+}$
is a slowly varying function, i.e. $\lim_{x\rightarrow\infty}\frac{L(\gamma x)}{L(x)}=1$ for any $\gamma>0$.

We assume that $B_{0}\in\mathcal{S}$ and either $\overline{B}\in\mathcal{R}(-\alpha-1)$ or $B\in\mathcal{G}$, i.e.
the maximum domain of attraction of Gumbel distribution $\Lambda(x)=\exp\{-e^{-x}\}$. 
A distribution function $F$ is in the maximal domain of attraction
of a distribution with distribution function $H(x)$ if there exist $a_{n}>0$, $b_{n}\in\mathbb{R}$ so that
\begin{equation}
\lim_{n\rightarrow\infty}n\overline{F}(a_{n}x+b_{n})=-\log H(x),\quad x\in\mathbb{R},
\end{equation}
where the limit is interpreted as $\infty$ when $H(x)=0$.
Therefore, the maximal domain of attraction of Gumbel distribution $\mathcal{G}$ consists of the distribution functions $F$ so that
there exist $a_{n}>0$, $b_{n}\in\mathbb{R}$ such that $\lim_{n\rightarrow\infty}n\overline{F}(a_{n}x+b_{n})=e^{-x}$, $x\in\mathbb{R}$.
$\mathcal{G}$ includes Weibull and lognormal distributions. 

$T_{i}=\tau_{i}-\tau_{i-1}$ is the length of the time interval between
two consecutive arrival times of the point process $\tau_{i-1}$ and $\tau_{i}$.
$\tau_{i}$ stands for the $i$th arrival time of the point process.
If $T_{i}$ are i.i.d., with mean $\mathbb{E}[T_{1}]$, then $N_{t}$ is a renewal process and
assume the usual net profit condition
\begin{equation}
\rho:=\frac{\mathbb{E}[C_{1}]}{p\mathbb{E}[T_{1}]}<1,
\end{equation}
then, it is well known that (see \cite{Teugels},
\cite{Veraverbeke} and \cite{Embrechts}), 
\begin{equation}
\lim_{u\rightarrow\infty}\frac{\psi(u)}{\overline{B}_{0}(u)}
=\frac{\rho}{1-\rho},\label{maineqn}
\end{equation}
where $\psi(u):=\mathbb{P}(\tau_{u}<\infty)$ is the infinite-horizon ruin probability, where
\begin{equation}
\tau_{u}:=\inf\{t>0:U_{t}\leq 0\}.
\end{equation}

The extensions when $N_{t}$ is not a renewal process has been studied in \cite{AsmussenEtAl}
when the authors consider a risk process with regenerative structures or a stationary and ergodic process satisfying
certain conditions. 
See also \cite{Araman}, \cite{Schlegel} and \cite{Zwart}. 

But in general, for a simple
point process $N_{t}$, we may not have a regenerative structure and it may not be stationary
and ergodic as assumed in \cite{AsmussenEtAl}. For example, none of the examples
that we will introduce later in Section \ref{Examples} are stationary or have a regenerative structure.
In this paper, we point out that the classical infinite-horizon ruin probability estimate \eqref{maineqn} 
and also finite-horizon ruin probabiliy estimate still hold
as long as there exists a large deviation principle for $(N_{t}/t\in\cdot)$, which is the main
result of this paper, i.e. Theorem \ref{mainthm} and Theorem \ref{Finite} in Section \ref{MainResults}. 
The intuition behind it is that if the arrival times deviate away from its mean with an exponentially
small probability, it will be dominated by the subexponential distributions of the claim sizes.
Our proof is essentially based on checking the conditions proposed in \cite{AsmussenEtAl}.

In Section \ref{Aggregate}, we review some known results about estimates of aggregate claims
when $N_{t}$ is not necessarily renewal and show that a condition is satisfied given the large deviation
principle of $(N_{t}/t\in\cdot)$.

Finally, in Section \ref{Examples}, we give three examples
of non-renewal processes, i.e. Hawkes process (which answers a question of
\cite{Stabile}), Cox process with shot noise intensity 
(which reproves a result that is known, see \cite{Asmussen}), and self-correcting point process
for which our results apply.

\section{Risk Process with Non-Renewal Arrivals and Regularly Varying Claims}

\subsection{Ruin Probabilites}\label{MainResults}

Before we proceed, recall that a sequence $(P_{n})_{n\in\mathbb{N}}$ of probability measures on a topological space $X$ 
satisfies the large deviation principle (LDP) with rate function $I:X\rightarrow\mathbb{R}$ if $I$ is non-negative, 
lower semicontinuous and for any measurable set $A$, we have
\begin{equation}
-\inf_{x\in A^{o}}I(x)\leq\liminf_{n\rightarrow\infty}\frac{1}{n}\log P_{n}(A)
\leq\limsup_{n\rightarrow\infty}\frac{1}{n}\log P_{n}(A)\leq-\inf_{x\in\overline{A}}I(x).
\end{equation}
Here, $A^{o}$ is the interior of $A$ and $\overline{A}$ is its closure. 
See \cite{Dembo} or \cite{VaradhanII} for general background regarding large deviations and the applications. 
Also \cite{Varadhan} has an excellent survey article on this subject.

The following assumption is the main assumption of this paper.

\begin{assumption}\label{A1}

(i) $(N_{t}/t\in\cdot)$
satisfies a large deviation principle with rate function $I(\cdot)$ such that $I(x)=0$ if and only if $x=\mu$.

(ii) $I(x)$ is increasing on $[\mu,\infty)$ and decreasing on $[0,\mu]$.

(iii) The net profit condition is satisfied,
\begin{equation}
\rho:=\frac{\mu\mathbb{E}[C_{1}]}{p}<1.
\end{equation}

(iv)
There exists some $\theta>0$ such that $\mathbb{E}[e^{\theta\sum_{i=1}^{n}T_{i}}]<\infty$ for any $n\in\mathbb{N}$.
\end{assumption}

Under Assumption \ref{A1}, the following two lemmas hold.

\begin{lemma}\label{L1}
Under Assumption \ref{A1}, for any fixed $\epsilon,\epsilon'>0$, there exists a constant $M>0$ such that
\begin{equation}
\mathbb{P}\left(\bigcap_{n=1}^{\infty}\left\{p\sum_{i=1}^{n}T_{i}\leq n\left(\frac{p}{\mu}+\epsilon\right)+M\right\}\right)>1-\epsilon'.
\end{equation}
\end{lemma}

\begin{proof}
Replacing $\epsilon$ by $p\epsilon$ and $M$ by $pM$ in the above equation, 
it is sufficient to prove that
\begin{equation}
\limsup_{M\rightarrow\infty}
\mathbb{P}\left(\bigcup_{n=1}^{\infty}\left\{\sum_{i=1}^{n}T_{i}>n\left(\frac{1}{\mu}+\epsilon\right)+M\right\}\right)=0
\end{equation}
Observe that $\{N_{t}\leq n\}=\{\sum_{i=1}^{n}T_{i}>t\}$ for any $n\in\mathbb{N}$ and $t\in\mathbb{R}^{+}$
and also for any fixed $\mu'<\mu$, there exists some $\delta'>0$ such that $I(\mu')-\delta'>0$
and for sufficiently large $t$,
\begin{equation}
\mathbb{P}(N_{t}/t<\mu')\leq e^{-t[I(\mu')-\delta']},
\end{equation}
where we used fact that $I(\mu')>0$ and $I(\cdot)$ is decreasing on $[0,\mu]$ from Assumption \ref{A1}.

Also for any $N\in\mathbb{N}$,
\begin{equation}
\limsup_{M\rightarrow\infty}\sum_{n<N}\mathbb{P}\left(\sum_{i=1}^{n}T_{i}>n\left(\frac{1}{\mu}+\epsilon\right)+M\right)=0.
\end{equation}
Together, take $N\in\mathbb{N}$ sufficiently large, 
\begin{align}
&\limsup_{M\rightarrow\infty}
\mathbb{P}\left(\bigcup_{n=1}^{\infty}\left\{\sum_{i=1}^{n}T_{i}>n\left(\frac{1}{\mu}+\epsilon\right)+M\right\}\right)
\\
&\leq\limsup_{M\rightarrow\infty}\sum_{n=1}^{\infty}
\mathbb{P}\left(\sum_{i=1}^{n}T_{i}>n\left(\frac{1}{\mu}+\epsilon\right)+M\right)\nonumber
\\
&=\limsup_{M\rightarrow\infty}\sum_{n\geq N}
\mathbb{P}\left(\sum_{i=1}^{n}T_{i}>n\left(\frac{1}{\mu}+\epsilon\right)+M\right)\nonumber
\\
&=\limsup_{M\rightarrow\infty}\sum_{n\geq N}
\mathbb{P}\left(\frac{N_{n(\mu^{-1}+\epsilon)+M}}{n(\mu^{-1}+\epsilon)+M}
\leq\frac{n}{n(\mu^{-1}+\epsilon)+M}\right)\nonumber
\\
&\leq\limsup_{M\rightarrow\infty}\sum_{n=1}^{\infty}
\mathbb{P}\left(\frac{N_{n(\mu^{-1}+\epsilon)+M}}{n(\mu^{-1}+\epsilon)+M}
\leq\frac{\mu}{1+\mu\epsilon}\right)\nonumber
\\
&\leq\limsup_{M\rightarrow\infty}\sum_{n\geq N}e^{-(n(\mu^{-1}+\epsilon)+M)[I(\frac{\mu}{1+\mu\epsilon})-\delta']}=0.\nonumber
\end{align}
\end{proof}

\begin{lemma}\label{L2}
Under Assumption \ref{A1} and further assume that $B_{0}\in\mathcal{S}$, 
\begin{equation}
\lim_{u\rightarrow\infty}\frac{\mathbb{P}(\sup_{n\geq 1}\{n(p\mu^{-1}-\epsilon)-p\sum_{i=1}^{n}T_{i}\}\geq u)}
{\overline{B}_{0}(u)}=0,
\end{equation}
for any sufficiently small $\epsilon>0$.
\end{lemma}

\begin{proof}
Notice that
\begin{align}
\mathbb{P}\left(\sup_{n\geq 1}\left\{n(p\mu^{-1}-\epsilon)-p\sum_{i=1}^{n}T_{i}\right\}\geq u\right)
&\leq\sum_{n=1}^{\infty}\mathbb{P}\left(\sum_{i=1}^{n}T_{i}\leq n\left(\frac{1}{\mu}-\frac{\epsilon}{p}\right)-\frac{u}{p}\right)
\\
&=\sum_{n>\frac{u}{\frac{p}{\mu}-\epsilon}}
\mathbb{P}\left(\sum_{i=1}^{n}T_{i}\leq n\left(\frac{1}{\mu}-\frac{\epsilon}{p}\right)-\frac{u}{p}\right)\nonumber
\\
&\leq\sum_{n>\frac{u}{\frac{p}{\mu}-\epsilon}}
\mathbb{P}\left(\sum_{i=1}^{n}T_{i}\leq n\left(\frac{1}{\mu}-\frac{\epsilon}{p}\right)\right)\nonumber
\\
&\leq\sum_{n>\frac{u}{\frac{p}{\mu}-\epsilon}}
\mathbb{P}\left(N_{n(\frac{1}{\mu}-\frac{\epsilon}{p})}\geq n\right)\nonumber
\\
&\leq\sum_{n>\frac{u}{\frac{p}{\mu}-\epsilon}}e^{-n(\frac{1}{\mu}-\frac{\epsilon}{p})[I((\frac{1}{\mu}-\frac{\epsilon}{p})^{-1})-\delta']}
,\nonumber
\end{align}
which is exponentially small in $u$ as $u\rightarrow\infty$. Since $B_{0}\in\mathcal{S}$ is subexponential,
we have the desired result.
\end{proof}

\cite{AsmussenEtAl} proved that \eqref{maineqn} holds if
we have Lemma \ref{L1} and Lemma \ref{L2}. So our main task here is to prove Lemma \ref{L1} and Lemma \ref{L2}
under following assumptions. Notice that Lemma \ref{L1} holds if $(T_{i})_{i\geq 1}$ is a stationary
and ergodic sequence (by using ergodic theorem). And that is the only place \cite{AsmussenEtAl}
used the stationarity and ergodicity assumption. That is why as long as we can prove Lemma \ref{L1} we can drop
the stationarity and ergodicity assumption. The following is the main assumption for the asymptotic results of ruin probabilities
that we are going to establish in this paper.

We have the following asymptotic estimates for infinite-horizon ruin probabilities.

\begin{theorem}\label{mainthm}
Under Assumption \ref{A1} and further assume that $B_{0}\in\mathcal{S}$, we have
\begin{equation}
\lim_{u\rightarrow\infty}\frac{\psi(u)}{\overline{B}_{0}(u)}
=\frac{\rho}{1-\rho}.
\end{equation}
\end{theorem}

\begin{proof}
It is a direct result of Lemma \ref{L1}, Lemma \ref{L2} and Theorem 3.1. in \cite{AsmussenEtAl}.
\end{proof}

\begin{remark}
In Theorem \ref{mainthm}, we can replace the large deviation assumption of $(N_{t}/t\in\cdot)$ by a 
large deviation assumption of $(\frac{1}{n}\sum_{i=1}^{n}T_{i}\in\cdot)$. But usually, if $(N_{t}/t\in\cdot)$
satisfies a large deviation principle with rate function $I(x)$ if and only if $(\frac{1}{n}\sum_{i=1}^{n}T_{i}\in\cdot)$
satisfies a large deviation principle with rate function $xI(1/x)$. The reason we chose to assume the large deviation
for $(N_{t}/t\in\cdot)$ in Assumption \ref{A1} is because when $N_{t}$ is not renewal, the inter-occurrence times
are not i.i.d. and it is usually easier and more natural to establish the large deviation for $(N_{t}/t\in\cdot)$, 
which is at least in the case of our three examples, Hawkes process, Cox process with shot
noise intensity and self-correcting point process.
\end{remark}

Next, let us consider the finite-horizon ruin probabilities.

Let $e(u):=\mathbb{E}[C_{1}-u|C_{1}>u]$ be the mean excess function and 
\begin{equation}
\psi(u,z):=\mathbb{P}(\tau_{u}\leq z),\quad z>0, 
\end{equation}
be the finite-horizon ruin probability.

\begin{remark}
(i) (Regularly Varying Distributions) 
If $\overline{B}(u)=\frac{L(u)}{u^{\alpha+1}}$, $\alpha\in(0,\infty)$, 
i.e. $\overline{B}\in\mathcal{R}(-\alpha-1)$, then, $e(u)\sim\frac{u}{\alpha}$.

(ii) (Lognormal Distributions) If $B(u)=\frac{1}{\sqrt{2\pi}}\int_{-\infty}^{(\log u-\mu)/\sigma}e^{-x^{2}/2}dx$,
then, $B\in\mathcal{G}$ and $B_{0}\in\mathcal{S}$ and $e(u)\sim\frac{\sigma^{2}u}{\log u-\mu}$.

(iii) (Weibull Distributions) If $B(u)=e^{-u^{\alpha}}$, where $\alpha\in(0,1)$, then, $B\in\mathcal{G}$ and
$B_{0}\in\mathcal{S}$ and $e(u)\sim\frac{u^{1-\alpha}}{\alpha}$.
\end{remark}

\begin{remark}
It is well known that if $B\in\mathcal{G}$, i.e. the maximal domain of attraction of Gumbel distribution, then,
\begin{equation}
\lim_{u\rightarrow\infty}\frac{\overline{B}(u+xe(u))}{\overline{B}(u)}=e^{-x},\quad x\in\mathbb{R}.
\end{equation}
\end{remark}

\begin{lemma}\label{ratio}
For any $y_{0}<\infty$, $\lim_{x\rightarrow\infty}\frac{\overline{G}(x+y)}{\overline{G}(x)}=1$ uniformly
for $y\in[0,y_{0}]$ for any $G\in\mathcal{S}$.
\end{lemma}

Lemma \ref{ratio} can be found in Chapter X of \cite{Asmussen}.

We have the following asymptotic estimates for finite-horizon ruin probabilities.

\begin{theorem}\label{Finite}
Under Assumption \ref{A1} and further assume that $B_{0}\in\mathcal{S}$, we have, for any $T>0$,
(i) If $\overline{B}\in\mathcal{R}(-\alpha-1)$, 
\begin{equation}
\lim_{u\rightarrow\infty}\frac{\psi(u,e(u)T)}{\overline{B}_{0}(u)}
=\frac{\rho}{1-\rho}\left[1-\left(1+(1-\rho)\frac{T}{\alpha}\right)^{-\alpha}\right].
\end{equation} 

(ii) If $B\in\mathcal{G}$,
\begin{equation}
\lim_{u\rightarrow\infty}\frac{\psi(u,e(u)T)}{\overline{B}_{0}(u)}
=\frac{\rho}{1-\rho}\left[1-e^{-(1-\rho)T}\right].
\end{equation}
\end{theorem}

\begin{proof}
The proof is based on the ideas in \cite{AsmussenEtAl} with some modifications.
When $N_{t}$ is a renewal process, \cite{AsmussenK} proved
both (i) and (ii). Now if $N_{t}$ satisfies Assumption \ref{A1}, then, by Lemma \ref{L1},
\begin{align}
\psi(u,e(u)T)
&=\mathbb{P}\left(\sup_{n\leq e(u)T}\left\{\sum_{i=1}^{n}C_{i}-p\sum_{i=1}^{n}T_{i}\right\}>u\right)
\\
&\geq(1-\epsilon')\mathbb{P}\left(\sup_{n\leq e(u)T}\left\{\sum_{i=1}^{n}C_{i}
-n\left(\frac{p}{\mu}+\epsilon\right)\right\}>u+M\right).\nonumber
\end{align}
Now, in both cases (i) and (ii), we know that $e(x)\sim\frac{\int_{x}^{\infty}\overline{B}(y)dy}{\overline{B}(x)}$.
Since both $B(x)$ and $B_{0}$ belong to $\mathcal{S}$,
Lemma \ref{ratio} implies that $\lim_{x\rightarrow\infty}\frac{e(x+y)}{e(x)}=1$ uniformly for $y\in[0,y_{0}]$ for any $y_{0}<\infty$.
Therefore, for any $\epsilon''\in(0,1)$, we have $e(u)\geq e(u+M)(1-\epsilon'')$ for any sufficiently large $u$ 
and thus we get
\begin{equation}
\psi(u,e(u)T)
\geq(1-\epsilon')\mathbb{P}\left(\sup_{n\leq e(u+M)T(1-\epsilon'')}\left\{\sum_{i=1}^{n}C_{i}
-n\left(\frac{p}{\mu}+\epsilon\right)\right\}>u+M\right).
\end{equation}

Now assume $\overline{B}\in\mathcal{R}(-\alpha-1)$. We have by the corresponding result for renewal $N_{t}$
in \cite{AsmussenK} and Lemma \ref{ratio}, 
\begin{align}
\liminf_{u\rightarrow\infty}\frac{\psi(u,e(u)T)}{\overline{B}_{0}(u)}
&=\liminf_{u\rightarrow\infty}\frac{\psi(u,e(u)T)}{\overline{B}_{0}(u+M)}
\\
&\geq(1-\epsilon')\frac{\rho_{\epsilon}}{1-\rho_{\epsilon}}\left[1-(1+(1-\rho_{\epsilon})T(1-\epsilon'')/\alpha)^{-\alpha}\right],
\nonumber
\end{align}
where $\rho_{\epsilon}:=\frac{\mathbb{E}[C_{1}]}{\frac{p}{\mu}+\epsilon}$. Since it holds for any $\epsilon,\epsilon',\epsilon''>0$,
we proved the lower bound. The case for $B\in\mathcal{G}$ is similar.

Now, let us prove the upper bound. Choose $\epsilon>0$ small enough that $\frac{p}{\mu}-\epsilon>\mathbb{E}[C_{1}]$,
\begin{align}
&\limsup_{u\rightarrow\infty}\frac{\psi(u,e(u)T)}{\overline{B}_{0}(u)}
\\
&=\limsup_{u\rightarrow\infty}\frac{\mathbb{P}(\sup_{n\leq e(u)T}\{\sum_{i=1}^{n}C_{i}-n(p\mu^{-1}-\epsilon)
+n(p\mu^{-1}-\epsilon)-\sum_{i=1}^{n}T_{i}\}>u)}
{\overline{B}_{0}(u)}\nonumber
\\
&\leq\limsup_{u\rightarrow\infty}\frac{\mathbb{P}(X_{\epsilon}(u)+Y_{\epsilon}>u)}{\overline{B}_{0}(u)},\nonumber
\end{align}
where $X_{\epsilon}(u):=\sup_{n\leq e(u)}\{\sum_{i=1}^{n}C_{i}-n(p\mu^{-1}-\epsilon)\}$
and $Y_{\epsilon}:=\sup_{n\geq 1}\{n(p\mu^{-1}-\epsilon)-\sum_{i=1}^{n}T_{i}\}$.
By Lemma \ref{L2}, we have $\lim_{u\rightarrow\infty}\frac{\mathbb{P}(Y_{\epsilon}>u)}{\overline{B}_{0}(u)}=0$
and by the results for the renewal case (\cite{AsmussenK}), for $\overline{B}\in\mathcal{R}(-\alpha-1)$,
\begin{equation}
\mathbb{P}(X_{\epsilon}(u)>u)\sim
\frac{\rho_{\epsilon}}{1-\rho_{\epsilon}}\left[1-(1+(1-\rho_{\epsilon})T/\alpha)^{-\alpha}\right]\overline{B}_{0}(u),
\end{equation}
where $\rho_{\epsilon}:=\frac{\mathbb{E}[C_{1}]}{p\mu^{-1}-\epsilon}$. Let us recall the Proposition 1.9. of 
Chapter X in \cite{Asmussen} which says
that for any distributions $A_{1},A_{2}$ on $\mathbb{R}^{+}$, if we have $\overline{A}_{i}(x)\sim a_{i}\overline{G}(x)$
for some $G\in\mathcal{S}$ and some constants $a_{1}+a_{2}>0$, 
then, $\overline{A_{1}\ast A_{2}}(x)\sim(a_{1}+a_{2})\overline{G}(x)$. In our case $G(x)=B_{0}(x)\in\mathcal{S}$
and $A_{1}$, $A_{2}$ are the distributions of $X_{\epsilon}(u)$ and $Y_{\epsilon}$ with $a_{1}>0$ and $a_{2}=0$.
Notice that $X_{\epsilon}(u)$ and $Y_{\epsilon}$ may be negative. To save the argument, we can 
simply use the fact that $X_{\epsilon}(u)\leq\max\{X_{\epsilon}(u),1\}$ and $Y_{\epsilon}\leq\max\{Y_{\epsilon},1\}$
then apply it to $\max\{X_{\epsilon}(u),1\}$ and $\max\{Y_{\epsilon},1\}$ instead. Also, in our case, $X_{\epsilon}(u)$ depends on $u$, 
but the proof of Proposition 1.9. Chapter X in \cite{Asmussen} still works. Hence, we get
\begin{equation}
\limsup_{u\rightarrow\infty}\frac{\psi(u,e(u)T)}{\overline{B}_{0}(u)}
\leq\frac{\rho_{\epsilon}}{1-\rho_{\epsilon}}\left[1-(1+(1-\rho_{\epsilon})T/\alpha)^{-\alpha}\right].
\end{equation}
Since it holds for any $\epsilon$, we proved the upper bound. The case for $B\in\mathcal{G}$ is similar.
\end{proof}

\subsection{Aggregate Claims}\label{Aggregate}

Let $A_{t}:=\sum_{i=1}^{N_{t-}}C_{i}$ be the aggregate claims up to time $t$, 
where as before we assume here that $C_{i}$ are i.i.d. positive random variables

Consider the following assumptions.
\begin{assumption}\label{A2}
(i) $\mathbb{E}[N_{t}]<\infty$ for any $t$ and $\mathbb{E}[N_{t}]\rightarrow\infty$ as $t\rightarrow\infty$. 

(ii) $\frac{N_{t}}{\mathbb{E}[N_{t}]}\rightarrow 1$, as $t\rightarrow\infty$.

(iii) There exist $\epsilon,\delta>0$ such that
\begin{equation}
\sum_{k>(1+\delta)\mathbb{E}[N_{t}]}
\mathbb{P}(N_{t}>k)(1+\epsilon)^{k}\rightarrow 0,
\end{equation}
as $t\rightarrow\infty$. 
\end{assumption}

\cite{Kluppelberg} proved that under Assumption \ref{A2}, for fixed time $t$, we have
\begin{equation}
\mathbb{P}(A_{t}-\mathbb{E}[A_{t}]>x)\sim\mathbb{E}[N_{t}]\mathbb{P}(C_{1}\geq x),\label{estimate}
\end{equation}
uniformly for $x\geq\gamma\mathbb{E}[N_{t}]$ for any $\gamma>0$.

\begin{remark} 
Indeed, \cite{Kluppelberg} proved a slightly stronger
result which says \eqref{estimate} holds assuming that the claim sizes $C_{i}$ are i.i.d.
with a distribution function $\overline{F}
\in ERV(-\alpha,-\beta)$ for some $1<\alpha\leq\beta<\infty$, where $ERV$ 
denotes the space of extended regular varying functions.
\end{remark}

It is usually easy to check (i) and also under the assumptions in Theorem \ref{mainthm},
$\frac{N_{t}}{t}\rightarrow\mu$ and $(N_{t}/t\in\cdot)$ satisfies a large deviation principle
with rate function $I(x)$ which is nonzero if and only if $x\neq\mu$.
Therefore, if we assume we could prove that $\frac{\mathbb{E}[N_{t}]}{t}\rightarrow\mu$ as $t\rightarrow\infty$,
then (ii) is satisfied. Moreover (iii) can be replaced by

(iii') For any $\mu'>0$, $c_{\mu'}:=\inf_{x\geq\mu'}\frac{I(x)}{x}>0$.

Assume (iii'), we can find some $0<\delta'<\delta$ such that for any $t$ sufficiently large,
\begin{align}
\sum_{k>(1+\delta)\mathbb{E}[N_{t}]}
\mathbb{P}(N_{t}>k)(1+\epsilon)^{k}
&\leq\sum_{k>(1+\delta')\mu t}
\mathbb{P}(N_{t}>k)(1+\epsilon)^{k}
\\
&\leq\sum_{k>(1+\delta')\mu t}e^{-(I(k/t)-\epsilon')t}(1+\epsilon)^{k}\nonumber
\\
&\leq\sum_{k>(1+\delta')\mu t}e^{-(I(k/t)\frac{t}{k}-\epsilon'\frac{1}{(1+\delta')\mu})k}(1+\epsilon)^{k}\nonumber
\\
&\leq\sum_{k>(1+\delta')\mu t}e^{-(c_{(1+\delta')\mu}-\epsilon'\frac{1}{(1+\delta')\mu})k}(1+\epsilon)^{k}.\nonumber
\end{align}
If we pick up $\epsilon'>0$ small enough such that $\epsilon'\frac{1}{(1+\delta')\mu}<c_{(1+\delta')\mu}$,
then, we can pick up $\epsilon>0$ small enough so that $c_{(1+\delta')\mu}-\epsilon'\frac{1}{(1+\delta')\mu}>\log(1+\epsilon)$
and therefore by letting $t\rightarrow\infty$, (iii) is satisfied.

\section{Examples of Non-Renewal Arrival Processes}\label{Examples}

\subsection{Example 1: Hawkes Process}

Hawkes process is a simple point process that has self-exciting property, clustering effect and long memory.
It was first introduced by \cite{Hawkes} and has been widely applied in finance, seismology, neuroscience, DNA modelling
and many other fields.
A simple point process $N_{t}$ is a linear Hawkes process if it has intensity
\begin{equation}
\lambda_{t}=\nu+\sum_{\tau<t}h(t-\tau),
\end{equation}
$h(\cdot):[0,\infty)\rightarrow(0,\infty)$ is integrable
and $\Vert h\Vert_{L^{1}}<1$. We also assume that $N_{t}$ starts with empty past history,
i.e. $N(-\infty,0]=0$. By our definition, the Hawkes process is non-stationary and is in general
even non-Markovian (unless $h(\cdot)$ is an exponential function). Also, it does not have a regenerative
structure. Thus, the conditions in \cite{Asmussen} do not apply here.

Notice that it is well known that, (see for example \cite{Daley})
\begin{equation}
\lim_{t\rightarrow\infty}\frac{N_{t}}{t}=\mu:=\frac{\nu}{1-\Vert h\Vert_{L^{1}}},
\end{equation}
and \cite{Bordenave} proved the a large deviation principle for $(N_{t}/t\in\cdot)$, i.e. Lemma \ref{BT}.
Therefore, it is natural that we can apply the results in our paper to study the ruin probabilities
with subexponential claims when the arrival process is a non-stationary linear Hawkes process.

\begin{lemma}[\cite{Bordenave}]\label{BT}
$(N_{t}/t\in\cdot)$ satisfies a large deviation principle with rate function,
\begin{equation}\label{ratefunctionI}
I(x)=
\begin{cases}
x\log\left(\frac{x}{\nu+x\Vert h\Vert_{L^{1}}}\right)-x+x\Vert h\Vert_{L^{1}}+\nu &\text{if $x\in[0,\infty)$}
\\
+\infty &\text{otherwise}
\end{cases}.
\end{equation}
\end{lemma}

\begin{remark}
Indeed, in \cite{Bordenave}, they expressed the rate function $I(\cdot)$ in an alternative way, which
is less explicit. The expression of the rate function in Lemma \ref{BT} was first pointed out in \cite{ZhuI}.
\end{remark}

\begin{lemma}\label{BT2}
$\frac{\mathbb{E}[N_{t}]}{t}\rightarrow\frac{\nu}{1-\Vert h\Vert_{L^{1}}}$ as $t\rightarrow\infty$.
\end{lemma}

\begin{proof}
Taking expectation of the indentity $\lambda_{t}=\nu+\int_{0}^{t}h(t-s)N(ds)$, we get
\begin{equation}
\mathbb{E}[\lambda_{t}]=\nu+\int_{0}^{t}h(t-s)\mathbb{E}[\lambda_{s}]ds\leq\nu+\Vert h\Vert_{L^{1}}\sup_{0\leq s\leq t}\mathbb{E}[\lambda_{s}]ds,
\end{equation}
which implies that for any $t$, $\sup_{0\leq s\leq t}\mathbb{E}[\lambda_{s}]\leq\frac{\nu}{1-\Vert h\Vert_{L^{1}}}$ and therefore
$\mathbb{E}[\lambda_{t}]\leq\frac{\nu}{1-\Vert h\Vert_{L^{1}}}$ uniformly in $t$. Next, let $H(t):=\int_{t}^{\infty}h(s)ds$ and
\begin{align}
\mathbb{E}[N_{t}]&=\mathbb{E}\left[\int_{0}^{t}\lambda_{s}ds\right]
\\
&=\nu t+\int_{0}^{t}\int_{0}^{s}h(s-u)d\mathbb{E}[N_{u}]ds\nonumber
\\
&=\nu t+\int_{0}^{t}\int_{u}^{t}h(s-u)dsd\mathbb{E}[N_{u}]\nonumber
\\
&=\nu t+\mathbb{E}[N_{t}]\Vert h\Vert_{L^{1}}-\int_{0}^{t}H(t-u)d\mathbb{E}[N_{u}],\nonumber
\end{align}
which implies that
\begin{equation}
\mathbb{E}[N_{t}]=\frac{\nu t}{1-\Vert h\Vert_{L^{1}}}-\int_{0}^{t}H(t-u)\mathbb{E}[\lambda_{u}]du,
\end{equation}
and
\begin{align}
\limsup_{t\rightarrow\infty}\frac{1}{t}\int_{0}^{t}H(t-u)\mathbb{E}[\lambda_{u}]du
&\leq\frac{\nu}{1-\Vert h\Vert_{L^{1}}}\limsup_{t\rightarrow\infty}\frac{1}{t}\int_{0}^{t}H(t-u)du
\\
&=\frac{\nu}{1-\Vert h\Vert_{L^{1}}}\limsup_{t\rightarrow\infty}\frac{1}{t}\int_{0}^{t}H(u)du=0,\nonumber
\end{align}
since $H(t)=\int_{t}^{\infty}h(s)ds\rightarrow 0$ as $t\rightarrow\infty$. 
\end{proof}

Assume the net profit condition $p>\mathbb{E}[C]\frac{\nu}{1-\Vert h\Vert_{L^{1}}}$.

If $C_{i}$ have light tails, then \cite{Stabile} obtained the asymptotics
for the infinite-horizon ruin probability $\psi(u)$ and the finite-horizon ruin probability $\phi(u,uz)$
for any $z>0$. As pointed out in \cite{Stabile} the case when $C_{i}$
are heavy-tailed is open and now we have the tools to handle the case.

\begin{proposition}\label{P1}
Assume the net profit condition $p>\mathbb{E}[C]\frac{\nu}{1-\Vert h\Vert_{L^{1}}}$.

(i) (Infinite-Horizon)
\begin{equation}
\lim_{u\rightarrow\infty}\frac{\psi(u)}{\overline{B}_{0}(u)}
=\frac{\nu\mathbb{E}[C_{1}]}{p(1-\Vert h\Vert_{L^{1}})-\nu\mathbb{E}[C_{1}]}.
\end{equation}

(ii) (Finite-Horizon) For any $T>0$,
\begin{align}
&\lim_{u\rightarrow\infty}\frac{\psi(u,uz)}{\overline{B}_{0}(u)}
\\
&=
\begin{cases}
\frac{\nu\mathbb{E}[C_{1}]}{p(1-\Vert h\Vert_{L^{1}})-\nu\mathbb{E}[C_{1}]}
\left[1-\left(1+\left(\frac{p(1-\Vert h\Vert_{L^{1}})-\nu\mathbb{E}[C_{1}]}{p(1-\Vert h\Vert_{L^{1}})}\right)
\frac{T}{\alpha}\right)^{-\alpha}\right]
&\text{if $\overline{B}\in\mathcal{R}(-\alpha-1)$}
\\
\frac{\nu\mathbb{E}[C_{1}]}{p(1-\Vert h\Vert_{L^{1}})-\nu\mathbb{E}[C_{1}]}
\left[1-e^{-\frac{p(1-\Vert h\Vert_{L^{1}})-\nu\mathbb{E}[C_{1}]}{p(1-\Vert h\Vert_{L^{1}})}T}\right]&\text{if $B\in\mathcal{G}$}
\end{cases}.\nonumber
\end{align}

(iii) (Aggregate Claims) For fixed time $t$,
\begin{equation}
\mathbb{P}(A_{t}-\mathbb{E}[A_{t}]>x)\sim\mathbb{E}[N_{t}]\mathbb{P}(C_{1}\geq x),
\end{equation}
uniformly for $x\geq\gamma\mathbb{E}[N_{t}]$ for any $\gamma>0$.
\end{proposition}

\begin{proof}
To prove (i) and (ii), by Theorem \ref{mainthm} and Theorem \ref{Finite}, it is enough to check the conditions in Assumption \ref{A1}.
(i) and (ii) of Assumption \ref{A2} can be verified by the large deviations result in Lemma \ref{BT} and the properties of the
rate function. (iii) of Assumption \ref{A1} is the assumption of the Proposition \ref{P1}. To check (iv) of Assumption \ref{A1},
notice that by the definition of Hawkes process,
$N_{t}$ stochastically dominates $N_{t}^{\nu}$, a homogenous Poisson process with parameter $\nu>0$. But $T_{i}^{\nu}$ corresponding to
$N_{t}^{\nu}$ are i.i.d. exponentially distributed with parameter $\nu$ and they stochastically dominate $T_{i}$, the length of time interval
between two consecutive arrivals of a Hawkes process. But we know that exponentially distribution has exponential tails and thus
for $\theta>0$ small enough, $\mathbb{E}[e^{\theta\sum_{i=1}^{n}T_{i}}]\leq\mathbb{E}[e^{\theta\sum_{i=1}^{n}T_{i}^{\nu}}]
=\mathbb{E}[e^{\theta T_{1}^{\nu}}]^{n}<\infty$ for any $n\in\mathbb{N}$. Thus (iv) of Assumption \ref{A1} holds. Now, to prove (iii),
it is enough to check (i), (ii) and (iii') of Assumption \ref{A2}. In the proof of Lemma \ref{BT2}, we showed
that $\mathbb{E}[\lambda_{t}]\leq\frac{\nu}{1-\Vert h\Vert_{L^{1}}}$ uniformly in $t$ and thus 
$\mathbb{E}[N_{t}]=\mathbb{E}\left[\int_{0}^{t}\lambda_{s}ds\right]\leq\frac{\nu t}{1-\Vert h\Vert_{L^{1}}}<\infty$ and (i) of Assumption \ref{A2}
is verified. (ii) of Assumption \ref{A2} is a result of Lemma \ref{BT2} and law of large numbers of $N_{t}/t$
and finally (iii') of Assumption \ref{A2} can be verified
by easily checking the rate function in Lemma \ref{BT}.
\end{proof}

\begin{remark}
Indeed, the large deviations for nonlinear Hawkes processes have been established in \cite{ZhuI} and \cite{ZhuII}. Unlike
linear Hawkes processes, the rate function for the large deviations for nonlinear Hawkes processes are less explicit and it is 
therefore more difficult to check if it satisfies the conditions in this paper. This has to be left for future investigations.
\end{remark}

\subsection{Example 2: Cox Process with Shot Noise Intensity}

We consider a Cox process $N_{t}$ with intensity $\lambda_{t}$ that follows a shot noise process
\begin{equation}
\lambda_{t}=\nu(t)+\sum_{\tau^{(1)}<t}g(t-\tau^{(1)}),
\end{equation}
where $\tau^{(1)}$ are the arrival times of an external homogenous Poisson process with intensity $\gamma$.
Here, $g(\cdot):\mathbb{R}^{+}\rightarrow\mathbb{R}^{+}$ is integrable, i.e. $\int_{0}^{\infty}g(t)dt<\infty$
and $\nu(t)$ is a positive, continuous, deterministic function such that $\nu(t)\rightarrow\nu$ as $t\rightarrow\infty$.

The ruin probabilities for heavy-tailed claims with arrival process being a shot noise Cox process is known
in the literature, e.g. see the book by \cite{Asmussen}. But the techniques
in the literature use the very specific features of shot noise Cox process and the proofs are much longer.
Our proof essentially only needs the large deviation result for $(N_{t}/t\in\cdot)$ which is very easy to establish.

Since $N^{(1)}$ is a Poisson process with intensity $\gamma$, by the definition of $\lambda_{t}$, it is easy
to see that
\begin{equation}
\frac{N_{t}}{t}\rightarrow\nu+\gamma\Vert g\Vert_{L^{1}},\quad\text{as $t\rightarrow\infty$}.
\end{equation}

It is not clear to the author if the large deviation result for $(N_{t}/t\in\cdot)$ is known in the literature.
For the sake of completeness, let us establish the large deviation principle here.

\begin{lemma}\label{SN}
$(N_{t}/t\in\cdot)$ satisfies a large deviation principle with rate function,
\begin{equation}
I(x)=
\begin{cases}
\sup_{\theta\in\mathbb{R}}\left\{\theta x-(e^{\theta}-1)\nu
-\gamma(e^{(e^{\theta}-1)\Vert g\Vert_{L^{1}}}-1)\right\}
&\text{if $x\in[0,\infty)$}
\\
+\infty &\text{otherwise}
\end{cases}.
\end{equation}
\end{lemma}

\begin{proof}
For any $\theta\in\mathbb{R}$, we have
\begin{align}
\mathbb{E}[e^{\theta N_{t}}]
&=\mathbb{E}\left[e^{(e^{\theta}-1)\int_{0}^{t}\lambda_{s}ds}\right]
\\
&=e^{(e^{\theta}-1)\int_{0}^{t}\nu(s)ds}\mathbb{E}\left[e^{(e^{\theta}-1)\int_{0}^{t}\int_{0}^{s}g(s-u)N^{(1)}(du)ds}\right]
\nonumber
\\
&=e^{(e^{\theta}-1)\int_{0}^{t}\nu(s)ds}\mathbb{E}\left[e^{\int_{0}^{t}\left[\int_{u}^{t}(e^{\theta}-1)g(s-u)ds\right]N^{(1)}(du)}
\right]\nonumber
\\
&=e^{(e^{\theta}-1)\int_{0}^{t}\nu(s)ds}e^{\gamma\int_{0}^{t}(e^{\int_{u}^{t}(e^{\theta}-1)g(s-u)ds}-1)du}\nonumber
\\
&=e^{(e^{\theta}-1)\int_{0}^{t}\nu(s)ds}e^{\gamma\int_{0}^{t}(e^{\int_{0}^{t-u}(e^{\theta}-1)g(s)ds}-1)du}\nonumber
\\
&=e^{(e^{\theta}-1)\int_{0}^{t}\nu(s)ds}e^{\gamma\int_{0}^{t}(e^{\int_{0}^{u}(e^{\theta}-1)g(s)ds}-1)du}.\nonumber
\end{align}
Therefore, we have
\begin{equation}
\lim_{t\rightarrow\infty}\frac{1}{t}\log\mathbb{E}[e^{\theta N_{t}}]
=(e^{\theta}-1)\nu+\gamma(e^{(e^{\theta}-1)\Vert g\Vert_{L^{1}}}-1).
\end{equation}
By G\"{a}rtner-Ellis theorem, we conclude that $(N_{t}/t\in\cdot)$ satisfies a large deviation principle
with rate function 
\begin{equation}
I(x)=
\sup_{\theta\in\mathbb{R}}\left\{\theta x-(e^{\theta}-1)\nu
-\gamma(e^{(e^{\theta}-1)\Vert g\Vert_{L^{1}}}-1)\right\}.
\end{equation}
Now, if $x<0$, then for any $\theta<0$, $\theta x-(e^{\theta}-1)\nu
-\gamma(e^{(e^{\theta}-1)\Vert g\Vert_{L^{1}}}-1)\geq\theta x\rightarrow\infty$ if we let $\theta\rightarrow-\infty$. Hence, $I(x)=+\infty$
for $x<0$.
\end{proof}

\begin{lemma}\label{SN2}
$\frac{\mathbb{E}[N_{t}]}{t}\rightarrow\nu+\gamma\Vert g\Vert_{L^{1}}$ as $t\rightarrow\infty$.
\end{lemma}

\begin{proof}
Observe that
\begin{align}
\mathbb{E}[N_{t}]
&=\mathbb{E}\left[\int_{0}^{t}\lambda_{s}ds\right]
\\
&=\int_{0}^{t}\nu(s)ds+\mathbb{E}\left[\int_{0}^{t}\int_{0}^{s}g(s-u)N^{(1)}(du)ds\right]\nonumber
\\
&=\int_{0}^{t}\nu(s)ds+\gamma\int_{0}^{t}\int_{0}^{s}g(s-u)duds\nonumber
\\
&=\int_{0}^{t}\nu(s)ds+\gamma\int_{0}^{t}\int_{0}^{s}g(u)duds,\nonumber
\end{align}
which implies that 
$\frac{\mathbb{E}[N_{t}]}{t}\rightarrow\nu+\gamma\Vert g\Vert_{L^{1}}$ as $t\rightarrow\infty$.
\end{proof}

\begin{proposition}\label{P2}
Assume the net profit condition $p>\mathbb{E}[C](\nu+\gamma\Vert g\Vert_{L^{1}})$.

(i) (Infinite-Horizon)
\begin{equation}
\lim_{u\rightarrow\infty}\frac{\psi(u)}{\overline{B}_{0}(u)}
=\frac{(\nu+\gamma\Vert g\Vert_{L^{1}})\mathbb{E}[C_{1}]}{p-(\nu+\gamma\Vert g\Vert_{L^{1}})\mathbb{E}[C_{1}]}.
\end{equation}

(ii) (Finite-Horizon) For any $T>0$,
\begin{align}
&\lim_{u\rightarrow\infty}\frac{\psi(u,uz)}{\overline{B}_{0}(u)}
\\
&=
\begin{cases}
\frac{(\nu+\gamma\Vert g\Vert_{L^{1}})\mathbb{E}[C_{1}]}{p-(\nu+\gamma\Vert g\Vert_{L^{1}})\mathbb{E}[C_{1}]}
\left[1-\left(1+\left(1-\frac{(\nu+\gamma\Vert g\Vert_{L^{1}})\mathbb{E}[C_{1}]}{p}\right)\frac{T}{\alpha}\right)^{-\alpha}\right]
&\text{if $\overline{B}\in\mathcal{R}(-\alpha-1)$}
\\
\frac{(\nu+\gamma\Vert g\Vert_{L^{1}})\mathbb{E}[C_{1}]}{p-(\nu+\gamma\Vert g\Vert_{L^{1}})\mathbb{E}[C_{1}]}
\left[1-e^{-(p-(\nu+\gamma\Vert g\Vert_{L^{1}})\mathbb{E}[C_{1}])T/p}\right]&\text{if $B\in\mathcal{G}$}
\end{cases}.\nonumber
\end{align}

(iii) (Aggregate Claims) For fixed time $t$,
\begin{equation}
\mathbb{P}(A_{t}-\mathbb{E}[A_{t}]>x)\sim\mathbb{E}[N_{t}]\mathbb{P}(C_{1}\geq x),
\end{equation}
uniformly for $x\geq\gamma\mathbb{E}[N_{t}]$ for any $\gamma>0$.
\end{proposition}

\begin{proof}
To prove (i) and (ii), by Theorem \ref{mainthm} and Theorem \ref{Finite}, it is enough to check the conditions in Assumption \ref{A1}.
(i) and (ii) of Assumption \ref{A2} can be verified by the large deviations result in Lemma \ref{SN} and the properties of the
rate function. (iii) of Assumption \ref{A1} is the assumption of the Proposition \ref{P2}. To check (iv) of Assumption \ref{A1},
notice that by the definition of Hawkes process,
$N_{t}$ stochastically dominates $N_{t}^{\nu^{\ast}}$, an homogenous Poisson process with parameter $\nu^{\ast}:=\max_{t\geq 0}\nu(t)$. 
But $T_{i}^{\nu^{\ast}}$ corresponding to
$N_{t}^{\nu^{\ast}}$ are i.i.d. exponentially distributed with parameter $\nu^{\ast}$ and they stochastically dominate $T_{i}$, the length of time interval
between two consecutive arrivals of a Hawkes process. But we know that exponentially distribution has exponential tails and thus
for $\theta>0$ small enough, $\mathbb{E}[e^{\theta\sum_{i=1}^{n}T_{i}}]\leq\mathbb{E}[e^{\theta\sum_{i=1}^{n}T_{i}^{\nu^{\ast}}}]
=\mathbb{E}[e^{\theta T_{1}^{\nu^{\ast}}}]^{n}<\infty$ for any $n\in\mathbb{N}$. Thus (iv) of Assumption \ref{A1} holds. Now, to prove (iii),
it is enough to check (i), (ii) and (iii') of Assumption \ref{A2}. It is easy to see that
that $\mathbb{E}[\lambda_{t}]=\nu(t)+\gamma\int_{0}^{t}g(s)ds<\infty$ for any $t>0$ and thus 
$\mathbb{E}[N_{t}]=\mathbb{E}\left[\int_{0}^{t}\lambda_{s}ds\right]<\infty$ and (i) of Assumption \ref{A2}
is verified. (ii) of Assumption \ref{A2} is a result of Lemma \ref{SN2} and law of large numbers of $N_{t}/t$
and finally (iii') of Assumption \ref{A2} can be verified
by easily checking the rate function in Lemma \ref{SN}.
\end{proof}

\subsection{Example 3: Self-Correcting Point Process}

A self-correcting point process, also known as the stress-release model, is a simple point process $N$
with empty history, i.e. $N(-\infty,0]=0$ such that it
admits the $\mathcal{F}_{t}$-intensity 
\begin{equation}
\lambda_{t}:=\lambda(Z_{t}),\quad\text{and}\quad
Z_{t}:=t-N_{t-}.\label{masterdefinition}
\end{equation}
The rate function
$\lambda(\cdot):\mathbb{R}\rightarrow\mathbb{R}^{+}$
is continuous and increasing such that
\begin{equation}
0<\lambda^{-}=\lim_{z\rightarrow-\infty}\lambda(z)<1<\lim_{z\rightarrow+\infty}\lambda(z)=\lambda^{+}<\infty.
\end{equation}
Notice that in the definition of intensity in \eqref{masterdefinition}, we used $N_{t-}$ instead of $N_{t}$. That is crucial
to guarantee that the intensity $\lambda_{t}$ for the self-correcting point process is $\mathcal{F}_{t}$-predictable.

The model was first introduced by \cite{Isham} as an example of a process that automatically 
corrects a deviation from its mean. Later, it was studied as a model in seismology. 
The stress builds up at the linear rate $1$ in our model and releases by the amount $1$ at $i$th jump. 
\cite{Vere} discussed an insurance interpretation.  

Under these assumptions, it is well known that $\frac{N_{t}}{t}\rightarrow 1$ as $t\rightarrow\infty$ 
(See for example Proposition 4.3 in \cite{Zheng}). 
Recently, \cite{SenandZhu} proved the following large deviation result.

\begin{lemma}[\cite{SenandZhu}]\label{SZ}
$(N_{t}/t\in\cdot)$ satisfies a large deviation 
principle with rate function
\begin{equation}
I(x)=
\begin{cases}
\Lambda^{-}(x) &\text{if $x>1$}
\\
0 &\text{if $x=1$}
\\
\Lambda^{+}(x) &\text{if $0\leq x<1$}
\\
+\infty &\text{otherwise}
\end{cases},
\end{equation}
where
\begin{equation}
\Lambda^{\pm}(x)=\log\left(\frac{x}{\lambda^{\pm}}\right)x+\lambda^{\pm}-x,\quad x\geq 0.
\end{equation}
\end{lemma}

\begin{lemma}\label{SZ2}
$\frac{\mathbb{E}[N_{t}]}{t}\rightarrow 1$ as $t\rightarrow\infty$.
\end{lemma}

\begin{proof}
$\mathbb{E}[N_{t}]=\mathbb{E}\left[\int_{0}^{t}\lambda(Z_{s})ds\right]$. Zheng \cite{Zheng} proved
that there exists a unique invariant measure $\pi(dz)$ for the Markov process $Z_{t}$. 
By ergodic theorem, we have
\begin{equation}
\frac{1}{t}\int_{0}^{t}\lambda(Z_{s})ds\rightarrow\int\lambda(z)\pi(dz),
\end{equation}
as $t\rightarrow\infty$.
We know that $Z_{t}=t-N_{t}$ has the generator
\begin{equation}
\mathcal{A}f(z)=\frac{\partial f}{\partial z}+\lambda(z)(f(z-1)-f(z)),
\end{equation}
and we have $\mathcal{A}z\pi=0$ which implies that $\int\lambda(z)\pi(dz)=1$ and thus
$\frac{1}{t}\int_{0}^{t}\lambda(Z_{s})ds\rightarrow 1$ a.s. as $t\rightarrow\infty$. 
Since $\lambda^{-}\leq\lambda(\cdot)\leq\lambda^{+}$, by bounded convergence theorem, we conclude
that $\frac{\mathbb{E}[N_{t}]}{t}\rightarrow 1$ as $t\rightarrow\infty$. 
\end{proof}

\begin{proposition}\label{P3}
Assume the net profit condition $p>\mathbb{E}[C]$.

(i) (Infinite-Horizon)
\begin{equation}
\lim_{u\rightarrow\infty}\frac{\psi(u)}{\overline{B}_{0}(u)}
=\frac{\mathbb{E}[C_{1}]}{p-\mathbb{E}[C_{1}]}.
\end{equation}

(ii) (Finite-Horizon) For any $T>0$,
\begin{equation}
\lim_{u\rightarrow\infty}\frac{\psi(u,uz)}{\overline{B}_{0}(u)}
=
\begin{cases}
\frac{\mathbb{E}[C_{1}]}{p-\mathbb{E}[C_{1}]}\left[1-\left(1+\left(1-\frac{\mathbb{E}[C_{1}]}{p}\right)\frac{T}{\alpha}\right)^{-\alpha}\right]
&\text{if $\overline{B}\in\mathcal{R}(-\alpha-1)$}
\\
\frac{\mathbb{E}[C_{1}]}{p-\mathbb{E}[C_{1}]}\left[1-e^{-(p-\mathbb{E}[C_{1}])T/p}\right]&\text{if $B\in\mathcal{G}$}
\end{cases}.
\end{equation}

(iii) (Aggregate Claims) For fixed time $t$,
\begin{equation}
\mathbb{P}(A_{t}-\mathbb{E}[A_{t}]>x)\sim\mathbb{E}[N_{t}]\mathbb{P}(C_{1}\geq x),
\end{equation}
uniformly for $x\geq\gamma\mathbb{E}[N_{t}]$ for any $\gamma>0$.
\end{proposition}

\begin{proof}
To prove (i) and (ii), by Theorem \ref{mainthm} and Theorem \ref{Finite}, it is enough to check the conditions in Assumption \ref{A1}.
(i) and (ii) of Assumption \ref{A2} can be verified by the large deviations result in Lemma \ref{SZ} and the properties of the
rate function. (iii) of Assumption \ref{A1} is the assumption of the Proposition \ref{P3}. To check (iv) of Assumption \ref{A1},
notice that by the definition of Hawkes process,
$N_{t}$ stochastically dominates $N_{t}^{\lambda^{-}}$, an homogenous Poisson process with parameter $\lambda^{-}$. 
But $T_{i}^{\lambda^{-}}$ corresponding to
$N_{t}^{\lambda^{-}}$ are i.i.d. exponentially distributed with parameter $\lambda^{-}$ and 
they stochastically dominate $T_{i}$, the length of time interval
between two consecutive arrivals of a Hawkes process. But we know that exponentially distribution has exponential tails and thus
for $\theta>0$ small enough, $\mathbb{E}[e^{\theta\sum_{i=1}^{n}T_{i}}]\leq\mathbb{E}[e^{\theta\sum_{i=1}^{n}T_{i}^{\lambda^{-}}}]
=\mathbb{E}[e^{\theta T_{1}^{\lambda^{-}}}]^{n}<\infty$ for any $n\in\mathbb{N}$. Thus (iv) of Assumption \ref{A1} holds. Now, to prove (iii),
it is enough to check (i), (ii) and (iii') of Assumption \ref{A2}. It is easy to see that
that $\lambda_{t}\leq\lambda^{+}<\infty$ for any $t>0$ and thus 
$\mathbb{E}[N_{t}]=\mathbb{E}\left[\int_{0}^{t}\lambda_{s}ds\right]\leq\lambda^{+}t<\infty$ and (i) of Assumption \ref{A2}
is verified. (ii) of Assumption \ref{A2} is a result of Lemma \ref{SZ2} and law of large numbers of $N_{t}/t$
and finally (iii') of Assumption \ref{A2} can be verified
by easily checking the rate function in Lemma \ref{SZ}.
\end{proof}

\section*{Acknowledgements}

The author is supported by NSF grant DMS-0904701, DARPA grant and MacCracken Fellowship at New York University.
The author is very grateful to an anonymous referee for the helpful comments and suggestions.

\bibliographystyle{elsarticle-harv}
\bibliography{mybib}

\end{document}